\documentclass[a4paper,UKenglish]{lipics-v2018}

\usepackage{microtype,amssymb} 
\usepackage{mdframed} 
\usepackage{hyperref}
\newtheorem{thm}{Theorem}
\newtheorem{lem}[thm]{Lemma}

\newtheorem{prp}[thm]{Proposition}

\def\vol{\mbox{vol}}
\def\volesti{{\tt VolEsti}}
\def\RR{{\mathbb R}}

\title{Practical volume computation of structured convex bodies, and an application to modeling portfolio dependencies and financial crises\footnote{The views expressed are those of the authors and do not necessarily reflect official positions of the European Commission.}}
\titlerunning{Volume computation of structured convex bodies}

\author{Ludovic Cal\`es\footnote{Cal\`es acknowledges the financial support of the European Commission through its Proof-of-Concept program.}}{European Commission, Joint Research Centre, Ispra, Italy}{ludovic.cales@ec.europa.eu}{}{}{}
\author{Apostolos Chalkis}{Department of Informatics \& Telecommunications\\
National \& Kapodistrian University of Athens, Greece}{achalkis@di.uoa.gr}{}{}{}
\author{Ioannis Z.~Emiris\footnote{Emiris is partially supported by the European Union’s Horizon 2020 research and innovation programme under grant agreement No 734242 (Project LAMBDA).}}{Department of Informatics \& Telecommunications\\
National \& Kapodistrian University of Athens, Greece, and \\ ATHENA Research \& Innovation Center, Greece}{emiris@di.uoa.gr}{}{}{}
\author{Vissarion Fisikopoulos}{Oracle, Greece}{vissarion.fysikopoulos@oracle.com}{}{}{}

\authorrunning{L.~Cal\`es et al.} 

\Copyright{Ludovic Cal\`es, Apostolos Chalkis, Ioannis Z.~Emiris and Vissarion Fisikopoulos}

\subjclass{Design and analysis of algorithms: \\Computational geometry, Random walks and Markov chains}
\keywords{Polytope volume, convex body, simplex, sampling, financial portfolio}

\category{}

\supplement{}

\funding{}

\acknowledgements{}

\EventEditors{Bettina Speckmann and Csaba D. T{\'o}th}
\EventNoEds{2}
\EventLongTitle{34th International Symposium on Computational Geometry (SoCG 2018)}
\EventAcronym{SoCG}
\EventYear{2018}
\EventDate{June 11--14, 2018}
\EventLocation{Budapest, Hungary}
\EventLogo{socg-logo} 
\SeriesVolume{99}
\ArticleNo{19} %
\nolinenumbers
\hideLIPIcs
\begin{document}

\maketitle

\begin{abstract}
We examine volume computation of general-dimensional polytopes and more general convex bodies, defined as the intersection of a simplex by a family of parallel hyperplanes, and another family of parallel hyperplanes or a family of concentric ellipsoids. Such convex bodies appear in modeling and predicting financial crises.
The impact of crises on the economy (labor, income, etc.) makes its detection of prime interest for the public in general and for policy makers in particular.  Certain features of 
dependencies in the markets clearly identify times of turmoil.
We describe the relationship between asset characteristics by means of a copula; 
each characteristic is either a linear or quadratic form of the portfolio components, hence the copula can be constructed by computing volumes of convex bodies. 

We design and implement practical algorithms in the exact and approximate setting, we experimentally juxtapose them and study the tradeoff of exactness and accuracy for speed.  We analyze the following methods in order of increasing generality:
rejection sampling relying on uniformly sampling the simplex, which is the fastest approach, but inaccurate for small volumes;
exact formulae based on the computation of integrals of probability distribution functions, which are the method of choice for intersections with a single hyperplane;
an optimized Lawrence sign decomposition method, since the polytopes at hand are shown to be simple with additional structure;
Markov chain Monte Carlo algorithms using random walks based on the hit-and-run paradigm generalized to nonlinear convex bodies and relying on new methods for computing a ball enclosed in the given body, such as a second-order cone program;
the latter is experimentally extended to non-convex bodies with very encouraging results.
Our C++ software, based on CGAL and Eigen and available on {\tt github}, is shown to be very effective in up to 100 dimensions.
Our results offer novel, effective means of computing portfolio dependencies and an indicator of financial crises, which is shown to correctly identify past crises.
\end{abstract}

\section{Introduction} 

\subsection{Financial context and motivation}

Modern finance has been pioneered by Markowitz who set a framework to study choice in portfolio allocation under uncertainty, see \cite{M52}.\footnote{for which he was awarded the Nobel Prize in economics in 1990.} 
Within this framework, Markowitz characterized portfolios by their return and their risk which is defined as the variance of the portfolios' returns.
And an investor would build a portfolio that will maximize its expected return for a chosen level of risk.\footnote{In the same way, by choosing a level of expected return, an investor can construct a portfolio which minimizes the risk.} 
It has since be common for asset managers to optimize their portfolio within this framework.
And it has led a large part of the empirical finance research to focus on the so-called efficient frontier which is defined as the set of portfolios presenting the lowest risk for a given expected return.
Figure \ref{fig:Illustration_Intro} (left panel) presents such an efficient frontier. The region on the left of the efficient frontier represent the portfolios domain.

Interestingly, despite the fact that this framework considers the whole set of portfolios, no attention has been given to the distribution of portfolios. Figure \ref{fig:Illustration_Intro} (middle panel) presents such distribution\footnote{10.000.000 portfolios have been sampled as presented later in Section \ref{Ssimplex}.}. 
When comparing the contour of the empirical portfolios distribution\footnote{Region over which at least 1 random portfolio lies.} and the portfolio domain bounded by the efficient frontier in Figure \ref{fig:Illustration_Intro} (right panel), we observe that the density of portfolios along the efficient frontier is dim and that most of the portfolios are located in a small region of the portfolios domain. 

\begin{figure}[h!]
	\centering
		\includegraphics[width=\textwidth]{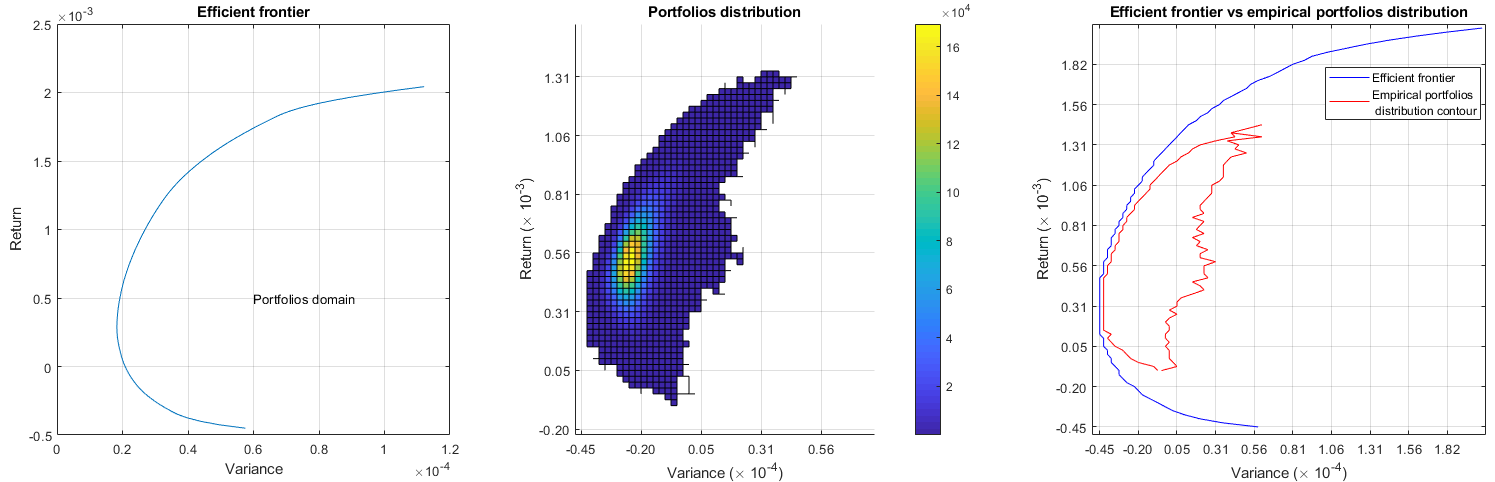}
	\caption{(left) Efficient frontier, (middle) Empirical portfolio distribution by portfolios' return and variance, (right) Efficient frontier in blue and contour of the empirical portfolio distribution in red. The market considered is made of the 19 sectoral indices of the DJSTOXX 600 Europe. The data is from October 16, 2017 to January 10, 2018.}
	\label{fig:Illustration_Intro}
\end{figure}

We also know from the financial literature that financial markets exhibit 3 types of behavior. 
In normal times, stocks are characterized by slightly positive returns and a moderate volatility, in up-market times (typically bubbles) by high returns and low volatility, and during financial crises by strongly negative returns and high volatility, see e.g. \cite{BGP12} for details.
So, following Markowitz' framework, in normal and up-market times, the stocks and portfolios with the lowest volatility should present the lowest returns, whereas during crises those with the lowest volatility should present the highest returns.
These features\footnote{also called ``stylized facts" in the financial literature} motivate us to describe the time-varying dependency between portfolios' returns and volatility.

However this dependency is difficult to capture from the usual mean-variance representation, as in Figure \ref{fig:Illustration_Intro} (middle panel), so we will rely on the copula representation of the portfolios distribution.
A copula is a bivariate probability distribution for which the marginal probability distribution of each variable is uniform.
As we following Markowitz' framework, the variables considered are the portfolios' return and variance.
Figure \ref{fig:Illustration_Copula} illustrates such a copula and shows a positive dependency between portfolios returns and variances.
Each line and column sum to 1\% of the portfolios. 

\begin{figure}[h!]
	\centering
		\includegraphics[width=0.5\textwidth]{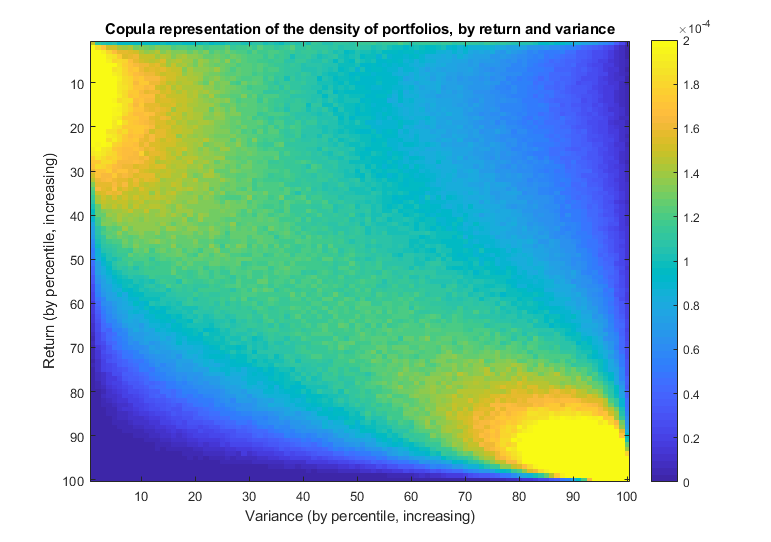}
	\caption{Copula representation of the portfolios distribution, by return and variance. The market considered is made of the 19 sectoral indices of the DJSTOXX 600 Europe. The data is from October 16, 2017 to January 10, 2018.}
	\label{fig:Illustration_Copula}
\end{figure}

The methods introduced here can be used to study other dependencies such as the momentum effect~\cite{JT93} which is implied by the dependencies of asset returns with their past returns.

The dependencies mentioned here are important because
\begin{itemize}
	\item through the return/volatility dependency, the detection of crisis raises policy makers awareness and allows them to act accordingly with potentially large implications in citizens' life (employment, wages, pensions, etc).
	\item the momentum, if persistent, questions the efficiency of financial markets, a strong assumption which still cannot be proven wrong.
\end{itemize}

Interestingly, the copulas can be computed over a single period of time making the information available as early as the sample allows. The copula for the momentum dependency can be computed over very short periods (even intra-daily). The copula for the return/volatility dependency requires the estimation of the stock returns variance-covariance matrix which has to be estimated over a sufficiently large period of time to be reliable thus delaying the detection of crisis.\footnote{Methods exist to estimate the stock returns variance-covariance matrix over short periods, see e.g. the range-based estimation method~\cite{BGP12}. However they usually requires high-frequency data and are not widely used. These methods are beyond the scope of this paper.}

In the general case, the framework to describe the dependencies is as follows.
First, as the set of portfolios, we consider the canonical $d$-dimensional simplex $\Delta^d\subset\RR^{d+1}$ where each point represents a portfolio and $d+1$ is the number of assets. 
The vertices represent portfolios composed entirely of a single asset.
The portfolio weights, i.e.\ fraction of investment to a specific asset, are non-negative and sum to 1. 
This is the most common investment set in practice today, as portfolio managers are typically forbidden from short-selling or leveraging. 
Second, considering some asset characteristic $ac$ quantified by $C\in\RR^{d+1}$, we define a corresponding quantity $f_{ac}(\omega, C)$ for any portfolio $\omega \in \Delta^d$.
For instance, considering the vector of asset returns $R\in\RR^{d+1}$, $\omega$ has the return $f_{ret}(\omega, R) = R^T \omega$.
Then, we define the cross-sectional score of a given portfolio $\omega^*$ as
$$
\rho_{ac} = \frac{ \vol(\Delta^*)}{\vol(\Delta^d)} ,\; \mbox{ where } 
\Delta^*=\{ \omega \in \Delta^d \, :\, f(\omega, C) \leq f(\omega^*, C)\} ,
$$ 
which corresponds to the share of portfolios with a return lower or equal to $R^*=R^T \omega^*$. 
This score corresponds to the cumulative distribution function of $f_{ac}(\omega,C)$ where the portfolios are uniformly distributed over the simplex.
In the following, we consider the cases where $f_{ac}$ is a linear combination or a quadratic form of $C$. Finally, the relationship between two asset characteristics $ac_1$ and $ac_2$ is presented in the form of a {\em copula} whose marginals are $\rho_{ac_1}$, $\rho_{ac_2}$. In our applications, the asset characteristics considered are the assets' returns and variances, and their values correspond to a linear combination of the returns and a quadratic form of the returns, respectively.

A copula is computed by slicing a simplex, i.e. the set of portfolios, along the asset characteristics. Thus, these questions are formulated in terms of convex bodies defined by intersecting simplices on one hand by a family of parallel hyperplanes and, on the other hand, by another family of parallel hyperplanes in the linear case or a family of concentric ellipsoids in the quadratic case. Furthermore, the latter case yields non-convex bodies between two ellipsoids.  

\subsection{Previous work}

The cross-sectional score of portfolio returns has been introduced in \cite{P05} and it is estimated by means of a quasi-Monte Carlo method. 
The applications have been limited in terms of dimensions: the 30 DAX components and the 24 MSCI Netherlands components in \cite{P05}, the 35 components of the IBEX in \cite{PST04}. This score has also been proposed in \cite{BCG11}, for the set of long/short equally weighted zero-dollar portfolios and whose estimation relying on combinatorics and statistics is computationally limited to around 20 dimensions, and in \cite{BH11} where the focus was not on a precise score.

Given that volume computation of polytopes is $\#$P-hard for both V- and H-representations \cite{DyerFr88} and no poly-time algorithm can achieve better than exponential error~\cite{Elekes86}, the problem is not expected to admit of an efficient deterministic algorithm in general dimension. 
Developing algorithms for volume computation has received a lot of attention in the exact setting~\cite{Bueler00}.
In the approximate setting, following the breakthrough polynomial-time algorithm by random walks~\cite{DyerFrKa91}, several algorithmic improvements ensued.
The current best theoretical bounds are in~\cite{Lee17b} and for polytope sampling in~\cite{Lee17a}.
Interestingly, only two pieces of software offer practical algorithms in high dimension:
\volesti, a public-domain C++ implementation that scales to a few hundred dimensions~\cite{VolEsti}, based on the Hit-and-Run paradigm~\cite{Lovasz99}, and the Matlab implementation of~\cite{CousinsV13}, which treats hyperplanes as an ellipsoid, and seems competitive to \volesti\ in very high dimensions.
Sampling from non-convex bodies appears in experimental works, with very few methods offering theoretical guarantees, e.g.\ in star shaped bodies \cite{Chandrasekaran09} or, more recently, in~\cite{Abbasi17}. 

\subsection{Our contribution}

We design and implement the following different approaches for volume computation:
Efficient sampling from the simplex and using rejection to approximate the target volume, which is fast but inaccurate for small volumes.
Exact formulae of integrals of appropriate probability distribution functions, which are implemented for the case of a single hyperplane.
Optimizing the use of Lawrence's sign decomposition method, since the polytopes at hand are shown to be simple with extra structure; a major issue here is numerical instability. 
Extending state-of-the-art random walks, based on the hit-and-run paradigm, to convex bodies defined as the intersection of linear halfspaces and ellipsoids.
The latter is experimentally generalized to non-convex bodies defined by two ellipsoids with same quadratic form, and accurate approximations are obtained under certain mild conditions.

Our randomized algorithms for volume approximation extend \volesti , where the main problem to address is to compute the maximum inscribed ball of the convex body $P$ a.k.a.\ Chebychev ball.
This reduces to a linear program when $P$ is a polytope.
For a convex body defined by intersecting a polytope with $k$ balls, the question becomes a second-order cone program (SOCP) with $k$ cones. When interchanging input balls with ellipsoids, the SOCP yields a sufficiently good approximation of the Chebychev ball.

Our implementations are in C++, lie in the public domain ({\tt github}),
are based on CGAL, rely on Eigen for linear algebra, on {\tt Boost} for random number generators, and experiment with two SOCP solvers for initializing random walks.  
Our software tools are general and of independent interest.
They are applied to allow us to extend the computation of a portfolio score to up to~100 dimensions, thus doubling the size of assets studied in financial research.
We thus provide a new description of asset characteristics dependencies.
Our methods allow us to propose and to effectively compute a new indicator of financial crises, which is shown to correctly identify all past crises with which we experimented.
More importantly, it allows us to establish that periods of momentum nearly never overlap with the crisis events, which is a new result in finance.

The rest of the paper is organized as follows. In Section \ref{finmod} we describe the convex bodies that intersect the canonical simplex and arise from the financial modeling we work on. Moreover we overview methods for representing and uniform sampling from the simplex and give theoretical guarantees for the sampling - rejection method accuracy.
Section \ref{Shplanes} considers volumes defined as the intersection of a simplex and one hyperplane or more hyperplanes, the latter being organized in at most two families of parallel hyperplanes.
Section \ref{Sellipse} studies convex and non-convex bodies defined as the intersection of a simplex and an ellipsoid, for which random walk methods are developed.
The implementations are discussed in Section \ref{Sexper}, along with experiments that show the validity of our approach in answering open questions in finance.
We conclude with current work and open questions.
Figures and tables that do not fit are given in the Appendix.

\section{Convex bodies and Financial modeling}\label{finmod}

We analyze real data consisting of regular interval (e.g.\ daily) returns of stocks such as the constituents of the Dow Jones Stoxx 600 Europe\texttrademark  (DJ600).
These are points in real space of dimension $d=600$, respectively: $r_i=(r_{i,1},\dots,r_{i,d}) \in \mathbb{R}^{d}$, $i\ge 1$.

We apply the methodology to a subset of assets drawn from the DJ~600 constituents\footnote{The data used is from Bloomberg\texttrademark.. It is daily and ranges from 01/01/1990 to 31/11/2017.}. 
Since not all stocks are tracked for the full period of time, we select the 100 assets with the longest history in the index\footnote{This implies a survivor bias, but we use it to assess the effectiveness of the methodology. One would wish to keep 600 constituents, replacing the exiting stocks with the entering ones along the sample.}, and juxtapose:
\begin{itemize}
	\item stock returns and stock returns covariance matrix over the same period to detect crises,
	\item stock returns and past stock returns to observe any momentum effect,
\end{itemize}

In financial applications, one considers compound returns over periods of $k$ observations, where typically $k=20$ or $k=60$; the latter corresponds to roughly 3 months when observations are daily.
Compound returns are obtained using $k$ observations starting at the $i$-th one where the $j$-th coordinate corresponds to asset $j$ and the component j of the new vector equals:
$$
(1+r_{i,j})(1+r_{i+1,j})\cdots (1+r_{i+k-1,j}) -1, \quad j=1,\dots, d.
$$
This defines the normal vector to a family of parallel hyperplanes, whose equations are fully defined by selecting appropriate constants. 
The second family of parallel hyperplanes is defined similarly by using an adjacent period of $k$ observations.

The covariance matrix of the stock returns is computed using the shrinkage estimator of \cite{LW04},\footnote{Matlab code on \url{http://www.econ.uzh.ch/en/people/faculty/wolf/publications.html}.} as it provides a robust estimate even when the sample size is short with respect to the number of assets.
A covariance matrix $C$ defines a family of ellipsoids centered at the origin $0\in\mathbb{R}^{d}$ whose equations $x^TCx=c$ are fully specified by selecting appropriate constants $c$.

To compute the copulas, we determine constants defining hyperplanes and ellispoids so that the volume between two consecutive such objects is 1\% of the simplex volume. 
The former are determined by bisection using the Varsi's exact formula.
For ellipsoids $E(x)=c_i$, we look for the $c_i$'s by sampling the simplex, then evaluating $E(x)$ at each point. The values are sorted and the $c_i$ selected so as to define intervals containing 1\% of the values. Two consecutive ellipsoids intersecting the simplex and the family of parallel hyperplanes define a non-convex body for which we practically extend \volesti \ algorithm.

The volume between two consecutive hyperplanes and two consecutive ellipsoids defines the density of portfolios whose returns and volatilities lie between the specified constants. We thus get a copula representing the distribution of the portfolios with respect to the portfolios returns and volatilities. 
Fig.~\ref{fig:Returns_variance_relationship} illustrates such copulae, and shows the different relationship between returns and volatility in good (left, dot-com bubble) and bad (right, bubble burst) times.\footnote{We consider 100 components of DJ~600 with longest history, over 60 days ending at the given date.}

The main problem is to compute all the volumes that arise from the intersection of the two families with the unit simplex. We totally have to handle three types of full dimensional bodies and thus we develop or use existing methods for three different problems. The first is to compute the volume of the polytope defined by the intersection of the unit simplex with four hyperplanes which are pairwise parallel. The second arises when an ellispoid intersects with the unit simplex and a family of parallel hyperplanes. The third is to compute the volume of a non-convex body defined by the intersection of two concentric ellipsoids with a simplex and a family of parallel hyperplanes.

We develop and use four methods in total. The first (M1) is an exact formula for the volume defined by the intersection of simplex with a hyperplane. The second (M2 or s/r) is to sample the unit simplex and approximate all the volumes directly. The third method (M3) is the optimized Lawrence formula for simple polytopes and is used for the first problem. The fourth method (M4) is the generalization of the \volesti \ algorithm to non-linear and non-convex bodies.

\subsection{Simplex representation and sampling}\label{Ssimplex}

This subsection sets the notation, surveys methods for uniform sampling of the simplex, and discusses their efficient implementation.

The $d$-dimensional simplex $\Delta^{d}\subset\RR^{d+1}$ may be represented by barycentric coordinates $\lambda=(\lambda_0,\dots,\lambda_d)$ s.t.\
$\sum_{i=0}^{d} \lambda_i=1, \, \lambda_i \ge 0$. The points are
$\sum_{i=0}^{d} \lambda_i v_i$, where $v_0, \dots, v_d \in \mathbb{R}^{d}$ are affinely independent.
It is convenient to use a full-dimensional simplex, by switching to Cartesian coordinates $x =(x_1, \dots, x_{d})$:
$$
m_{bc}: \mathbb{R}^{d+1} \mapsto \mathbb{R}^{d} : \lambda \rightarrow x = M (\lambda_1,\dots,\lambda_n)^T + v_0,\; \mbox{ where } M = [ v_1 -v_0\, \cdots \ v_{d} -v_0 ] ,
$$
is a $d\times d$ invertible matrix. The inverse transform is:
\begin{equation}\label{eq:cb}
m_{cb}: \mathbb{R}^{d} \mapsto \mathbb{R}^{d+1} : x \rightarrow \lambda = \left[\begin{array}{c} - 1_{d}^T \\ I_d \end{array}\right] M^{-1} (x-v_0) + \left[\begin{array}{c} 1 \\ 0_{d} \end{array}\right] ,
                \end{equation}
where $0_{d}, 1_{d}$ are $d$-dimensional column vectors of 1's and 0's, respectively, and $I_{d}$ is the $d$-dimensional identity matrix.

A number of algorithms exist for sampling, where some have been rediscovered, while others contain errors; see the survey~\cite{SmTr}.
Let us start with a unit simplex in Cartesian coordinates.
A $O(d \log d)$ algorithm is the following \cite{David81,D86,RbKr07}:
Generate $d$ distinct integers uniformly in $\{1,\dots ,K-1\}$,
where $K$ is the largest representable integer.
Sort them as follows: $x_0=0<x_1<\cdots <x_{d+1}=K$.
Now ${(x_i-x_{i-1})}/K$, $i=1,\dots ,d$, defines a uniform point.
Assuming we possess a perfect hash-function, the choice of distinct integers takes $O(d)$. 
For $d>60$ we implement a variant of Bloom filter to guarantee distinctness.

A linear-time algorithm is given in \cite{RbMel98}, which is is generally the algorithm of choice, although it is slower for $d<80$:
\begin{enumerate}
\item
Generate $d+1$ independent unit-exponential random variables $y_i$ by uniformly sampling real value $x_i\in (0,1)$ and setting $y_i=- \log x_i$.
\item
Normalize the $y_i$'s by their sum $s=\sum_{i=0}^d y_i$, thus obtaining a uniformly distributed point $(y_0/s,\dots,y_d/s)$ on the $d$-dimensional canonical simplex lying in $\RR^{d+1}$.
\item
Project this point along the $x_0$-axis to $(y_1/s,\dots, y_d/s)$, which is a uniform point in the full-dimensional unit simplex.
\end{enumerate}

To sample an arbitrary simplex, we can map sampled points from the unit simplex by transformation~(\ref{eq:cb}), which preserves uniformity. 
Due to applying the transformation, the complexity is $O(d^2)$ to generate a uniform point.
The same complexity, though slower in practice, is achieved in \cite{grimme}.

Sampling could be used in order to approximate all the volumes that arising when two families of parallel hyperplanes intersect with a simplex. One can sample the simplex and count the percentage of points in the region of interest. The complexity is 
$O(kd)$ to generate $k$ points.
In the case of a family of $\ell$ parallel hyperplanes, all sample points are evaluated at the hyperplane linear polynomials in time $O(kd)$.
Given the $\ell$ constant terms characterizing the hyperplanes, for each point we perform a binary search so as to decide in which layer it lies.
Hence the total complexity is $O(k\log \ell)$, which is dominated since $\ell\le 100$ typically.
Given a family of $\ell$ ellipsoids with same quadratic form intersecting a simplex, the method requires $O(kd^2)$ to evaluate all sample points and $O(k\log \ell)$ to assign them to layers.

\subsection{Sampling - Rejection accuracy}

In this subsection we obtain how we could guarantee a bounded error for the sampling - rejection method with high probability. Let $B$ be a convex or non convex full dimensional body in dimension $d$ and let $S$ be an enclosing simplex such that $B\subseteq S$ and let $p=\dfrac{Vol(B)}{Vol(S)}$. Then if we uniformly sample a point from $S$ it lies in $B$ with probability $p$. So if we sample $N$ points from $S$ the random variable $X$ which gives the number, $k$, of points that lie in $B$ follows the binomial distribution. So, $P(X=k)={{N}\choose{k}} p^k(1-p)^{N-k}$. Moreover if we sample $N$ points and reject,
\begin{equation}\label{poserr}
\sum_{k=n_1}^{n_2} P(X=k),\quad n_1=Np(1-e)\text{,  }n_2=Np(1+e)
\end{equation}
is the probability that the sampling - rejection method error is at most $e$. From Poison Limit Theorem we know that if $\lim_{N\rightarrow \infty}Np=\lambda$ is a constant independent of $N$ then for any fixed $k$, $lim_{N\rightarrow\infty}P(X=k)=e^{-\lambda}\dfrac{\lambda^k}{k!}$. If we set $N=m_1\cdot 10^x$, where $x=m_2+\lceil -\log_{10}p\rceil$, we notice that $Np\approx m_1\cdot 10^{m_2}$. So we can use \textit{Poisson limit theorem} to approximate the random variable $X$ as we usually set large enough values for $m_1$ and $m_2$. Then from Poison cumulative distribution function we can approximate probability \ref{poserr}.
\begin{table}[!h] 
\makebox[1 \textwidth][c]{ \resizebox{0.5 \textwidth}{!}{   
\begin{tabular}{|p{1.0cm}||p{2.7cm}|p{0.8cm}| }
 \hline
 \multicolumn{3}{|c|}{Sampling - Rejection error} \\ \hline
 $error$ & \centering $N$ & $Pr$ \\ \hline
$1\%$ & $4\cdot 10^{4+\lceil -\log_{10}p\rceil}$ & $0.955$\\
\hline
$2\%$ & $9\cdot 10^{3+\lceil -\log_{10}p\rceil}$ & $0.942$\\
\hline
$3\%$ & $4\cdot 10^{3+\lceil -\log_{10}p\rceil}$ & $0.942$\\
\hline
$4\%$ & $4\cdot 10^{3+\lceil -\log_{10}p\rceil}$ & $0.972$\\
\hline
$5\%$ & $2\cdot 10^{3+\lceil -\log_{10}p\rceil}$ & $0.975$\\
\hline
$6\%$ & $1\cdot 10^{3+\lceil -\log_{10}p\rceil}$ & $0.942$\\
\hline
$7\%$ & $8\cdot 10^{2+\lceil -\log_{10}p\rceil}$ & $0.952$\\
\hline
$8\%$ & $6\cdot 10^{2+\lceil -\log_{10}p\rceil}$ & $0.951$\\
\hline
$9\%$ & $5\cdot 10^{2+\lceil -\log_{10}p\rceil}$ & $0.956$\\
\hline
$10\%$ & $4\cdot 10^{2+\lceil -\log_{10}p\rceil}$ & $0.955$\\
\hline
\end{tabular}
} } 
\caption{\label{tableer} Maximum sampling-rejection method errors with high probability. $N$ is the number of points we have to sample.}
\end{table}

In Table \ref{tableer} we give for several values of errors the $m1,m2$ in order to get probabilities higher than 0.94. Notice that we use the order of $p$ at the exponent. In practice we could estimate $p$ while we are sampling and obtain the right order of $p$ with very high probability.

\section{Intersection with hyperplanes}\label{Shplanes}

This section considers computing the volume of the intersection of a simplex and one or more linear halfspaces. The most general case is to be given two families of parallel hyperplanes and consider all created polytopes.
We assume that the simplex is given in V-representation, i.e.\ as a set of vertices, and the hyperplanes by their equations. 

We can always transform the simplex to be a unit full-dimensional simplex with the origin as one vertex by the transformation of Sect.~\ref{Ssimplex}. The same transform applies to the hyperplanes, and volume ratios as preserved.

\subsection{Single halfspace formula}\label{SSformula}

Surprisingly, there exist an exact, iterative formula (M1) for the volume defined by intersecting a simplex with a hyperplane. A geometric proof is given in \cite{VARSI}, by subdividing the polytope into pyramids and, recursively, to simplices.
We implement a somewhat simpler formula~\cite{Ali73}, which also requires $O(d^2)$ operations.
Let $H = \{(x_1,\dots ,x_d) | \sum_{i=1}^{d} a_ix_i\leq z\}$ be the linear halfspace. 
\begin{enumerate} \item{Compute $u_j=a_j-z$, $j=1,\dots, d$. Label the nonnegative $u_j$ as $Y_1,\dots ,Y_K$ and the negatives as $X_1,\dots ,X_J$.}
{Initialize $A_0=1, A_1=A_2=\cdots =A_K=0$.}
\item For $h=1,2,\dots ,J$ repeat:
$A_k\longleftarrow \dfrac{Y_kA_k-X_hA_{k-1}}{Y_k-X_h}$, for $k=1,2,\dots ,K$.
\end{enumerate}
If $\Delta^d\subset\RR^d$ is the unit simplex then, for $h =J$, $A_K= {\vol(\Delta^d\cap H)} / {\vol(\Delta^d)}$.

Recall, from Section~\ref{Ssimplex}, that sampling uniformly over the simplex can be obtained by drawing exponential random variables. Thus, an alternative formula follows from computing the cumulative distribution of a linear combination of exponential random variables.
In \cite{Mathai}, they propose an exact method to compute the distribution $f$ of such linear combination. It consists in representing $f$ as its moment generating function, analogous to a Laplace transform, simplifying it with a generalized partial-fraction technique of integration, before inverting its terms. However, in double precision, the method showed numerical discrepancies above 20 dimensions and was thus abandoned. However, it has the advantage of being generalizable to nonlinear combinations (see Sect.~\ref{Sconcl}).  

\subsection{Simple polytopes}\label{SimPol}

This section considers simple polytopes defined by a constant number of families of parallel hyperplanes; in our application there are two such families.
The defined polytopes are simple, i.e., all vertices are defined at the intersection of $d$ hyperplanes, assuming that no hyperplane contains any of the simplex vertices and, moreover, two hyperplanes do not intersect on a simplex edge at the same point.

For a simple polytope $P$, the decomposition by Lawrence~\cite{Lawrence} picks $c \in \mathbb{R}^d$, $q\in\mathbb{R}$ such that $c^T x + q$ is not constant along any edge, i.e.\ $c,-c$ do not lie on the normal fan of any edge.
For each vertex $v$, let $A(v)$ be the $d \times d$ matrix whose columns correspond to the equations of hyperplanes through $v$.  Then $A(v)$ is invertible and vector $\gamma(v)$ such that $A(v) \gamma(v) = c$ is well defined up to a permutation. The assumption on $c$ assures no entry vanishes, then
$$
\vol(P) = \frac1{d!} \,
\sum_{v} \frac{ (c^T v + q )^ d }{ |\det A(v) | \; \prod_{i=1}^d \gamma(v)_i }.
$$
The computational complexity is $O (d^3 n)$, where $n$ is the number of vertices.
We set $q=0$ for simplicity in the implementation.
An issue is to choose $c$ so as to avoid that $c^T x + q$ be nearly
constant on some edge, because this would result in very small entries 
in the denominator and numerical issues. 
A theoretical choice is given in \cite{Lawrence}, but its practical importance is very small.
The main drawback of Lawrence's decomposition remains numerical instability when executed with floating point numbers, and high bit complexity, when executed over rational arithmetic. The latter is indispensible for $d>30$ in our applications, because then numerical results become very unstable.

To compute the volume defined by the intersection of a simplex and two arbitrary hyperplanes, we exploit the fact that the simplex is unit in order to compute more effectively the determinants and the solutions of the linear system.
The hardest case is when vertex $v$ is defined by the two arbitrary hyperplanes $H_a, H_b$, the supporting hyperplane $H_{0}:\sum_{i=1}^{d}x_i=1$, and $d-3$ hyperplanes of the form $H_i:x_i=0$. Then, up to row permutations, 
\begin{equation}\label{Ematrix} 
 A(v) \, = \, \begin{bmatrix} 
   -1 & \quad & \quad & 1 & a_1 & b_1 \\
   \quad & \ddots & \quad & \vdots & \vdots & \vdots \\
  \quad & \quad & -1 & 1 & a_{d-3} & b_{d-3} \\  
  \quad & \quad & \quad & 1 & a_{d-2} & b_{d-2} \\
  \quad & \quad & \quad & 1 & a_{d-1} & b_{d-1} \\
  \quad & \quad & \quad & 1 & a_{d} & b_{d} 
\end{bmatrix} ,
\end{equation}

where the $i_j,\, i=a,b$ are the coefficients of the equation of $H_i$ up to permutation.
Then we solve the lowest right $3\times 3$ linear system in $O(1)$ and then the computation of each remaining unknown $\gamma(v)_i$, $i=1,\dots ,d-3$ requires $O(1)$ operations for a total of $O(d).$
The corresponding determinant is computed in $O(1)$. 

\begin{lem}\label{lem1}
Polytopes in H-representation, defined by intersecting the simplex with two arbitrary hyperplanes in $\RR^d$, have $O(d^2)$ vertices, which are computed in $O(1)$ each.
\end{lem}
\begin{proof} 
A vertex in the new polytope is of one of 3 types:
(i) It may be a vertex of unit simplex $\Delta$. It suffices to check all simplex vertices against hyperplanes $H_a,H_b$ in total time $O(d)$.
(ii) It may be the intersection of a simplex edge with $H_a$, which is easy to identify and compute by intersecting simplex edges whose vertices lie on different sides of $H_a$, with $H_a$. Each such edge is defined by at least one coordinate hyperplane, so computing the edge intersection with $H_a$ is in $O(1)$. These vertices are checked against $H_b$ in $O(1)$ each, since they contain at most two nonzero coordinates. There are $O(d^2)$ such edges, hence the total complexity is $O(d^2)$.

(iii) It may be defined as $H_a\cap H_b\cap \Delta$, i.e.\ the intersection of $H_a$ with the edges of $H_b\cap \Delta$. Let $B_1,B_2$ be vertices on $H_b\cap \Delta$. Then $B_1$ is defined by the intersection of $H_b$ and an edge $(v_i,v_j)$ of the unit simplex, when $v_i$ and $v_j$ lie on different sides of $H_b$ and $B_2$ by the intersection of $H_b$ and an edge $(v_k,v_m)$. That means that every vertex in $H_b\cap \Delta$ corresponds to a unit simplex edge. Then we have 3 cases:
\begin{enumerate}
\item{$B_1,B_2$ lie on the same side of $H_a$: no vertex is defined.}
\item{If $i\neq k, i\neq m, j\neq k, j\neq m$ there is not an edge between $B_1$ and $B_2$.}
\item{If $B_1,B_2$ correspond to simplex egdes that have a common vertex 
and lie on different sides of $H_a$, then a polytope's vertex is defined, which has at most 3 nonzero coordinates.}
\end{enumerate}
In the worst case $d/2$ simplex vertices lie on the same side of $H_b$ and $d/2$ on the other. Then the polytope's vertices that are defined by $H_a\cap H_b\cap\Delta$ are at most $d\dfrac{d}{2}=O(d^2)$.
\end{proof}

Lawrence's formula requires both H- and V-representation. In our setting, the H-representation is known, but the previous lemma allows us to obtain vertices as well.

\begin{prp}
Let us consider polytopes defined by intersecting the simplex with two arbitrary hyperplanes.  
The total complexity of the Lawrence sign decomposition method, assuming that the H-representation is given, is $O(d^3)$.
\end{prp}

The entire discussion extends to polytopes defined by two families of parallel hyperplanes. The matrices $A(v)$ remain of the same form because each vertex is incident to at most one hyperplane from each family.

\section{Intersection with ellipsoids}\label{Sellipse}

This section considers more general convex bodies, defined as a finite, bounded intersection of linear and nonlinear halfspaces.
For this, we extend the polynomial-time approximation algorithm in \volesti~\cite{VolEsti} so as to handle nonlinear constraints.
Our primary motivation here is computing the volume of the intersection of a simplex with an ellipsoid in general dimension.

\subsection{Random walks}

The method in~\cite{VolEsti} follows the Hit-and-Run algorithm in~\cite{Lovasz99}, and is based on an approximation algorithm in $O^*(d^5)$.
It scales in a few hundred dimensions by integrating certain algorithmic improvements to the original method.
We have to generalize the method because the input is not a polytope but a general convex body, while \volesti \ works for $d$-polytopes.
It suffices to solve two subproblems: Compute the maximum inscribed ball of the convex body a.k.a.\ Chebychev ball, and compute the intersection points of a line that crosses the interior of the convex body $P$ with the boundary of $P$.

The first problem is treated in the next subsection.
For the second one, when the body is the intersection of linear and quadratic halfspaces, it suffices to solve systems of linear or quadratic equations. In our case where $P$ has few input hyperplanes we can optimize that procedure by transforming a base of our polytope to an orthonormal base thus obtaining very simple linear systems.
One heuristic is to first compute the intersection of the line with all hyperplanes and test whether the intersection points lie inside the ellipsoid so as to avoid intersecting the line with the ellipsoid.
Formally, every ray $\ell$ in Coordinate Direction Hit-and-Run is of the form ${p}+\lambda{e}_k$ and parallel to $d-1$ simplex facets. The roots of $\lambda^2+2\lambda p_k+|p|^2-R^2$ define the intersection of a sphere with radius $R$, centered at the origin, and a coordinate direction ray $\ell$.
If $C$ is the matrix of an ellipsoid centered at the origin its intersections with $\ell$ are roots of:
\begin{align*}
& C_{kk} x^2+bx+c=0,\quad
  b=2C_{kk}p_k + 2\sum_{j=k+1}^{d} C_{kj}p_j + 2\sum_{i=0}^{k-1} C_{ik}p_i,\\
& c=\sum_{i=0}^{d} C_{ii}p_{i}^{2} + 2\sum_{j=i+1}^{d} C_{ij}p_i p_j,\quad i=0,\dots, d.
\end{align*}
Computing the roots, and keeping the largest negative and smallest positive $\lambda$ is quite fast.

In our application, there are non-convex bodies defined by the intersection of two parallel hyperplanes and two concentric ellipsoids.  
We thus modify \volesti\ in order to compute the non convex volume. We make two major changes.
First, in ray shooting, we have to check whether one quadratic equation has only complex solutions, which implies the ray does not intersect the ellipsoid.  For $\lambda$, we take the largest negative and the smallest positive root in every step as well.
Second, for the initial interior point, we sample from the unit simplex and when we find a point inside the intersection we stop and use it for initialization. We define an inscribed ball with this center and radius equal to some small $\epsilon>0$. We stop the algorithm when we find the first inscribed ball as described in the next subsection. So we can set $\epsilon$ sufficiently small so it always defines an inscribed ball in practice, but the enclosing ball is enough to run the algorithm and do not stop until we find an inscribed ball. 

The method works fine for $d<35$ using the same walk length and number of points as for the convex case, and has time complexity and accuracy competitive to running \volesti \ on the convex set defined by one ellipsoid. For $d> 35$, the method fails to approximate volume for most of the cases. This should be due to inaccurate rounding bodies and the inscribed ball we define. 
Given these first positive results, various improvements are planned.

\subsection{Chebychev ball}\label{ChebBall}

This section offers methods for computing a ball inside the given convex region. Ideally, this is the largest inscribed ball, aka Chebychev ball, but a smaller ball may suffice.
Computing the Chebychev ball reduces to a linear program when $P$ is a polytope (p. 148 in \cite{Boyd_opt}). For general convex regions, more general methods are proposed.

At the very least, one point must be obtained inside the convex region.
When we do not have the Chebychev ball, an issue is that concentric balls with largest radii will again be entirely contained in the convex region, thus wasting time in the computation.
In practice we use the one interior point as center of an enclosing ball, then reduce the radius until the first inscribed ball.
To decide whether a given ball is inscribed, with high probability, we check whether all boundary points in Hit-and-Run belong to the sphere instead of any other constraint.

We start with some simple approaches.
Let us consider the case of intersecting a simplex with an ellipsoid.
If there are $z_1$ simplex vertices inside the ellipsoid and $z_2$ outside, then we have $(z_2+1)z_1$ vertices on the boundary of the convex intersection. Since $z_1+z_2=d+1$, then $(z_2+1)z_1\geq d+1$ and a new inscribed simplex is defined. In this case we take its largest inscribed ball and start hit and run.
More generally, we sample from the unit simplex until we have $d+1$ points inside our section and then take the largest inscribed ball of this new simplex that is defined by the $d+1$ points.
Another approach is to consider the transformation mapping the ellipsoid to a sphere and apply it both to the simplex and to the ellipsoid. We compute the distance from the sphere's center to the new simplex and compare it with the sphere's radius. 

For a convex body that comes from intersecting a polytope with $k$ balls the problem becomes a Second-Order Cone Program (SOCP) with $k$ cones.
However in our case we need to consider input ellipsoids.
Assume that we transformed the ellipsoid to a ball $B'=\{x'_c+u' \, :\, \|u'\|\leq r'\}$, and applied the same transformation to the simplex to have $a_i x\leq b_i$ for $i\in [d+1], a_i\in \mathbb{R}^d, b_i\in\mathbb{R}$.
The following SOCP computes the maximum ball $B=\{x_c+u\ :\ \|u\|\leq r\}$ in the intersection of the simplex and $B'$:
\begin{align*}
\max\ r,\quad \text{subject to}: a_i^Tx_c + r||a_i|| \leq b_i,\, ||x'_c - x_c|| \leq r'-r .
\end{align*}
There are several ways to solve SOCP's such as to reformulate it to as a semidefinite program or perform a quadratic program relaxation.
Moreover, since in our case we only have a single cone we could utilize special methods as in~\cite{ErIy}.
However, for our case it suffices to use the generic SOCP solver from~\cite{DoCh} as it is very efficient; for a random simplex and ball, it takes $0.06$~sec in $d=100$ and $<20$~sec in $d=1000$, on Matlab using {\tt ecos} and {\tt yalmip} packages.

It is possible to apply the inverse transformation and get an inscribed ellipsoid, which is not necessarily largest possible.
However we can use the maximum inscribed ball in that ellipsoid as an approximation of the Chebychev ball, by taking the center of that ellipsoid and the minimum eigenvalue of its matrix as the radius.

\subsection{Market volatility expressed by ellipsoids}

In our financial application, portfolios are points in the unit $d$-dimensional simplex $\Delta^{d}\subset \RR^{d+1}$ defined as the convex hull of $v_0,\dots,v_d\in\RR^d$, where $v_i$ lies on the $i$-th axis.
The simplex lies in hyperplane $\sum_{i=0}^d \lambda_i=1$.
To model levels of volatility, a family of full-dimensional ellipsoids in $\RR^{d+1}$, centered at the origin, is defined by the covariance matrix $C$ of asset returns.
We wish to compute the volume of intersections of this family with the simplex and, moreover, with a family of hyperplanes on the simplex. Rejection sampling would work in this context, however methods employing random walks require a full-dimensional convex body.
Given a full $(d+1)$-dimensional ellipsoid $G:\lambda^TC\lambda-c\leq 0$ centered at the origin, where $C\in\RR^{(d+1) \times (d+1)}$ is symmetric positive-definite, we compute the equation of the ellipsoid defined $G\cap \Delta^d\subset \RR^d$, by imposing the constraint $\sum_{i=0}^d \lambda_i=1$ by 
transform $m_{cb}$ in expression~(\ref{eq:cb}), thus obtaining:
$$
(x-v_0)^T \left( M^{-T} \, [-1\, I_d] \, C \left[ \begin{array}{c} 1\\ 0_d\end{array} \right] M^{-1} \right) (x-v_0) + A(x-v_0) = c',
$$
where the expression in parenthesis is the matrix defining the new $d$-dimensional ellipsoid in Cartesian coordinates, and $A\in\RR^{d \times d}, c'\in\RR$ are obtained by direct calculation.
Similarly the simplex maps to Cartesian coordinates.  

\section{Implementation and experiments} \label{Sexper} 

Our implementations are in C++, lie in the public domain\footnote{https://github.com/TolisChal/volume\_approximation}, and are using CGAL and Eigen.
All experiments of the paper have been performed on a personal computer with Intel Pentium G4400 3.30GHz CPU and 16GB RAM.
Times are averaged over 100 runs.
Some resulting tables and figures are given in the Appendix.


We test the following convex bodies: a $d$-simplex intersected with:
(1) two arbitrary halfspaces,
(2) two parallel halfspaces,
(3) an ellipsoid,
(4) two parallel halfspaces and two cocentric ellipsoids (non convex body).

In general, M1 is preferred when available. Method M2 is the fastest and scales easily to 100 dimensions, so it is expected to be useful for larger dimensions. However, for small volumes its accuracy degrades; sampling more points makes it slower than M4. The latter is thus the method of choice for volumes $<1\%$ of the simplex volume, but it is not clear whether it would be fast beyond $d=100$. Method M3 is useful, even for small volumes, but it cannot scale to $d=100$ due to numerical instability; if we opt for exact computing, it becomes too slow.

\subsection{Synthetic data}

The formula M1 is used in all Tables where exact computation is needed between two parallel hyperplanes intersecting the simplex.

Table~\ref{fig:text1} considers the intersection of an arbitrary simplex with two hyperplanes. The vertices of each simplex are randomly chosen uniformly from the surface of a ball with radius 100, using {\tt CGAL} random point generator. All hyperplanes' coefficients are randomly chosen in $[-10,10]$ with {\tt Boost} (mt19937) random generator.
For \volesti\ we do not use the rounding option for the input polytope. This means that skinny polytopes have low accuracy since the random walk mixes slow, cf.\ row 10 of Table~\ref{fig:text1}. On the other hand, M2 is not affected by polytope shape. 
Up to $d= 30$ and for large volume ratio, namely $>1\%$, M2 yields very accurate and fastest results. The last two experiments show that \volesti\ achieves the most accurate approximation when the ration of accepted sampled points is small.

In Table~\ref{fig:text2} we use same runtime for M2 and M4 (analogous numbers of sampled points) and compare their accuracy. We perform two experiments per dimension. For the first, for each dimension we compute the volume between two parallel hyperplanes which is $1\%$ of the simplex volume. For exact volume computation we use (M1). For the second experiment, for each dimension we compute volumes defined by the intersection of 4 hyperplanes which are pairwise parallel with the simplex, which is close to $0.01\%$ of the simplex volume. For exact computation we used {\tt vinci} default method, {\tt rlass}. M2 is fast but inaccurate for small volumes; M4 is most accurate but should not scale beyond $d=100$.

In Table~\ref{fig:text3} we have an arbitrary simplex and two arbitrary hyperplanes that intersect with it. 
We compare our Lawrence implementation in Sect.~\ref{SimPol}, using floating-point and rational computation, with {\tt rlass} and M2. We have two parallel hyperplanes intersect the unit simplex. {\tt vinci} fails to compute the volume for $d>31$.
Our exact computation works even in $d=100$ but becomes very slow.

Table~\ref{fig:text4} compares M2 (s/r) with two variants of M4 for ellipsoid intersection. The only difference for the latter is the way we construct an inscribed ball:
In s/V we implement random sampling until $d+1$ points are found, and in o/V we use SOCP. We see M2 is significantly faster than either variant of M4. All methods yield similar output values.
 
Table~\ref{fig:text5} compares s/r with Hit-and-Run for non-convex bodies, as in Sect.~\ref{Sellipse}. Very small values of Volume means the method failed to approximate the volume.
\subsection{Financial modeling with real data}

When we work with real data in order to build the indicator, we wish to compare the densities of portfolios along the two diagonals. In normal and up-market times, the portfolios with the lowest volatility present the lowest returns and the mass of portfolios should be on the up-diagonal. During crisis the portfolios with the lowest volatility present the highest returns and the mass of portfolios should be on the down-diagonal, see  Fig.~\ref{fig:Returns_variance_relationship}  as illustration.
Thus, setting up- and down-diagonal bands, we define the indicator as the ratio of the down-diagonal band over the up-diagonal band, discarding the intersection of the two. The construction of the indicator is illustrated in Fig.~\ref{fig:Figure_Crisis_Indicator} where the indicator is the ratio of the mass of portfolios in the blue area over the mass of portfolios in the red one.
\begin{figure}[h!]
\includegraphics[width=0.5\textwidth]{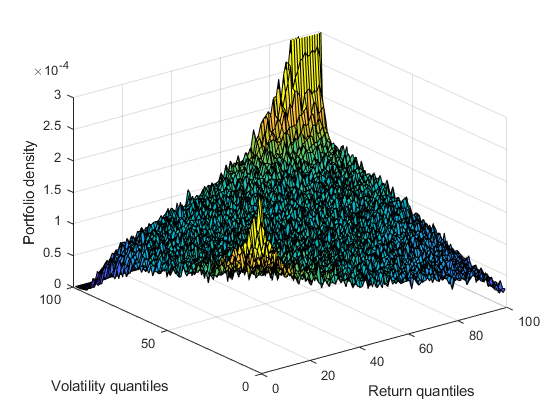} 
\includegraphics[width=0.5\textwidth]{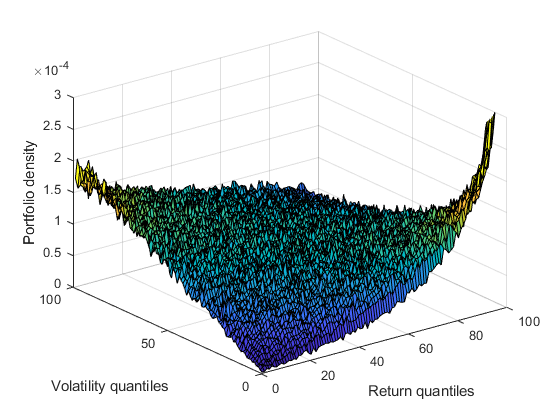} 
	\caption{Returns/variance relationship on the $1^{st}$ September 1999 (left), i.e. during the dot-com bubble, and on the $1^{st}$ September 2000 (right), at the beginning of the bubble burst. Blue= low density of portfolios, yellow=high density of portfolios.}
	\label{fig:Returns_variance_relationship} 
\end{figure}

\begin{figure}[h!]
	\centering
		\includegraphics[width=0.5\textwidth]{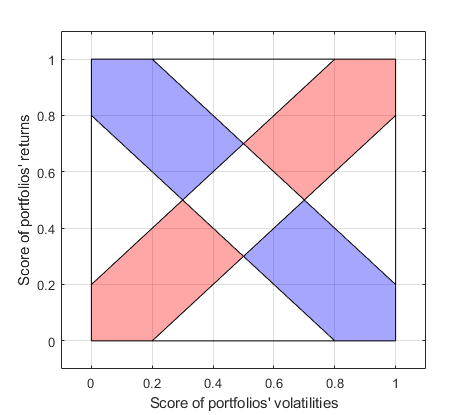}
	\caption{Illustration of the diagonal bands considered to build the indicator.}
	\label{fig:Figure_Crisis_Indicator}
\end{figure}

In the following, the indicator is computed using copulae estimated using the sampling method, drawing 500000 points. Computing the indicator over a rolling window of $k=60$ days and with a band of $\pm 10\%$ with respect to the diagonal, we report in Table \ref{tab:Warnings} all the periods over which the indicator is greater than 1 for more than 60 days. The periods should be more than 60 days to avoid the detection of isolated events whose persistence is only due to the auto-correlation implied by the rolling window. All these periods offer warnings, but only the longest ones correspond to crises. 

We compare these results with the database for financial crises in European countries proposed in ~\cite{ESRB17}. The first crisis (from May 1990 to Dec. 1990) corresponds to the early 90's recession, the second one (from May 2000 to May 2001) to the dot-com bubble burst, the third one (from Oct. 2001 to Apr. 2002) to the stock market downturn of 2002, the fourth one (from Nov. 2005 to Apr. 2006) is not listed and it is either a false signal or it might be due to a bias in the companies selected in the sample, and the fifth one (from Dec. 2007 to Aug. 2008) to the sub-prime crisis.

Regarding the momentum effect, i.e.\ the effect of the compound returns of the last 60 days on the following 60-day compound returns, we report the indicator in Fig~\ref{fig:Momentum}. We observe that there were only 10 events of lasting momentum effect, mostly around the 1998-2004 period. We remark that they nearly never overlap with the crisis events, with the exception of the end of 2011. To the authors' knowledge, this result is new in finance.

\section{Conclusion and Future work}\label{Sconcl}

Since runtimes are very reasonable, we plan to extend our study to larger subsets of assets of DJ~600 and eventually the whole index in $d=600$.
Another extension is to consider polytopes defined by intersections of both families of parallel hyperplanes and the family of ellipsoids, thus creating 3-D diagrams of dependencies, which have never been studied in finance: one difficulty is to model the outcome since visualization becomes intricate.

Random sampling follows a Monte Carlo (MC) approach by relying on C/C++ functions such as {\tt random} which implement pseudorandom generators.
We are experimenting with quasi-MC generators which require fewer points to simulate the uniform distribution. Our preliminary experiments indicate that this may yield a speedup of about~2.
An obvious enhancement is to parallelize our algorithms, which seems straightforward. Then results can be obtained for larger classes of assets such as the entire DJ~600.

One challenge is to extend the volume formula to the intersection with a ellipsoid.
In \cite{MP13}, they propose a method to approximate the distribution $f$ of quadratic forms in gamma random variables which is a similar problem to that in~\cite{Mathai} (see Sect.~\ref{Shplanes}). It consists in fitting $f$ with a generalized gamma distribution by matching its first 3 moments with those of $f$ and to adjust the distribution with a polynomial in order to fit the higher moments. To get an approximation with a polynomial of degree $k$, the method requires the first $2k$ moments.
In the case of a quadratic form in $d$ random variables, the moment of order $m$ is obtained by a sum over all the partitions of $m$ into $d^2$ terms. The number of partitions makes the computation of moments challenging even for $d \geq 5$.

\bibliography{p19-cales}
\newpage
\appendix

\section{Experiments} 

\begin{table}[!h]
\makebox[1 \textwidth][c]{       
\resizebox{1.3 \textwidth}{!}{   
\begin{tabular}{ |p{0.3cm}||p{0.8cm}|p{0.8cm}|p{0.25cm}||p{1.2cm}|p{1.7cm}|p{1.7cm}||p{1.7cm}|p{1.1cm}|p{1.9cm}|p{0.9cm}||p{0.6cm}|p{1.0cm}|p{0.7cm}| }
 \hline
 \multicolumn{14}{|c|}{Experimental results for arbitrary simplex and two arbitrary hyperplanes} \\
 \hline
 $d$ & $k$ & $m$ & n & $R/r$ & Simplex Vol. & {\tt Vinci} Vol. & s/r Vol. &
s/r Error & \volesti\ Vol. & \volesti\ Error & s/r Time & \volesti\ Time & {\tt Vinci} Time
\\
\hline
5 & $2\cdot 10^5$ & 1464 & 4 & 39.3877 & 225859 & 1638.14  & 1653.29 & 0.0092 & 1551.961  & 0.0526 & 0.076 & 1.189  & 0.0\\
\hline
5 & $2\cdot 10^5$ & 8269 & 5 & 240.261 & 31545.7 & 1287.54  & 1304.26 & 0.0130 & 1104.214  & 0.1423 & 0.072 &  2.474 & 0.0\\
\hline
10 & $3\cdot 10^5$ & 111018 & 7 & 31.1786 & 1.14352e+09 & 4.22648e+08  & 4.2317e+08 & 0.0012 & 4.399476e+08  & 0.0409 & 0.156 & 8.290  & 0.0\\
\hline
10 & $3\cdot 10^5$ & 16279 & 9 & 752.594 & 2.21485e+07 &  1.20556e+06 & 1.20185e+06 & 0.0031 &  0.023537e+06 & 0.9805 & 0.164 & 19.980  & 0.0\\
\hline
15 & $3\cdot 10^5$ & 1695 & 11 & 112.756 & 2.87936e+10 &  1.62617e+08 & 1.62684e+08 & 0.0004 &  1.284843e+08 & 0.2099 & 0.204 &  43.547 & 0.0\\
\hline
15 & $3\cdot 10^5$ & 168639 & 10 & 51.9497 & 1.8289e+11 &  1.02984e+11 & 1.02808e+11 & 0.0017 & 1.018419e+11  & 0.0111 & 0.224 & 31.848  & 0.0\\
\hline
20 & $4\cdot 10^5$ & 52657 & 17 & 50.351 & 2.47765e+17 & 3.24630e+16 & 3.26163e+16 & 0.0047 & 3.201464e+16 & 0.0138 & 0.416 & 135.685 & 0.0\\
\hline
20 & $4\cdot 10^5$ & 13952 & 17 & 140.094 & 6.76692e+15 & 2.38561e+14 & 2.3603e+14 & 0.0106 & 2.334992e+14 & 0.0212 & 0.42 & 181.058 & 0.0\\
\hline
25 & $4\cdot 10^5$ & 4982 & 23 & 135.804 & 1.37457e+18 & 1.70146e+16 & 1.71202e+16 & 0.0062 & 1.119995e+16 & 0.3417 & 0.52 & 333.052 & 0.0\\
\hline
25 & $4\cdot 10^5$ & 3809 & 25 & 89.8112 & 4.17323e+18 & 4.03833e+16 & 3.97396e+16 & 0.0159 & 5.017313e+18 & 123.2 & 0.508 & 384.346 & 0.0\\
\hline
30 & $4\cdot 10^5$ & 118304 & 22 & 4164.1 & 1.28638e+17 & 4.12910e+16  & 4.10773e+16 & 0.0052 & 5.02297e+16 & 0.2165 & 0.64 & 863.056 & 11.4\\
\hline
30 & $4\cdot 10^5$ & 27523 & 24 & 177.613 & 4.08094e+18 & 2.80038e+17  & 2.80799e+17 & 0.0027 & 1.891857e+17 & 0.3244 & 0.616 & 622.995 & 7.3\\
\hline\hline
10 & $3\cdot 10^5$ & 1151 & 10 & 61.3936 & 2.99231e+08 & 1.17756e+06  & 1.14805e+06 & 0.0251 & 1.185146e+06  & 0.0064 & 0.152 & 10.367  & 0.0\\
\hline
18 & $4\cdot 10^5$ & 1318 & 16 & 57.0641 & 8.58015e+11 & 2.96758e+09 & 2.82716e+09 & 0.0473 & 2.908083e+09 & 0.0200 & 0.376 & 93.7450 & 0.0\\
\hline
\end{tabular} 
} }
\caption{\label{fig:text1} $k$ is the number of points sampled in the unit simplex, $k=10^5\log d$, $m$ the number of points in the intersection, $n$ the number of vertices in the intersection. $R/r$ is the ratio of radii of the smallest enclosing over the largest inscribed ball of the simplex;
s/r is sampling with rejection; Error denotes relative error $(V-v)/V$ of computed value $v$ over exact volume $V$. Time is in seconds.
}
\end{table}

\quad
\begin{table}[!h]
\makebox[1 \textwidth][c]{       
\resizebox{1.3 \textwidth}{!}{   
\begin{tabular}{ |p{0.3cm}|p{0.8cm}|p{1.0cm}|p{1.5cm}|p{0.8cm}|p{1.0cm}|p{0.3cm}|p{0.4cm}|p{1.5cm}|p{1.0cm}|p{1.5cm}|p{0.8cm}| }
 \hline
 \multicolumn{12}{|c|}{Experimental results for rejection and \volesti.} \\
 \hline
 $d$ & $k$ & $m$ & s/r Vol & s/r time & $N$ & $W$ & $\epsilon$ & \volesti & Time \volesti & Exact Vol & Exact Time \\
 \hline
15 & 3$\cdot 10^7$ & 300345 & 7.66e-15 & 14.716 & 101551 & 11 & 0.4 & 7.52e-15 & 20.86 & 7.65e-15 & 0.0 \\
\hline
15 & 3$\cdot 10^7$ & 744 & 1.90e-17 & 14.796 & 101551 & 11 & 0.4 & 2.15e-17 & 21.49 & 2.01e-15 & 0.0 \\
\hline
20 & 3$\cdot 10^7$ & 299842 & 4.11e-21 & 23.532 & 66571 & 12 & 0.6 & 4.44e-21 & 36.17 & 4.11e-21 & 0.0 \\
\hline
20 & 3$\cdot 10^7$ & 2040 & 2.80e-23 & 23.688 & 66571 & 12 & 0.6 & 2.72e-23 & 34.88 & 2.74e-23 & 0.1 \\
\hline
25 & 3$\cdot 10^7$ & 299976 & 6.44e-28 & 30.74 & 50294 & 12 & 0.8 & 5.81e-28 & 34.03 & 6.45e-28 & 0.0 \\
\hline
25 & 3$\cdot 10^7$ & 980 & 2.10e-30 & 30.664 & 65691 & 12 & 0.7 & 2.00e-30 & 46.03 & 1.985e-30 & 0.1 \\
\hline
30 & 4$\cdot 10^7$ & 400395 & 3.77e-35 & 51.104 & 50388 & 13 & 0.9 & 3.42e-35 & 48.71 & 3.77e-35 & 0.0 \\
\hline
30 & 4$\cdot 10^7$ & 4769 & 4.49e-37 & 60.32 & 63772 & 13 & 0.8 & 4.52e-37 & 63.91 & 4.56e-37 & 3.2 \\
\hline
\end{tabular}
} }
\caption{\label{fig:text2} $k$ is the number of points sample from the unit simplex, $k=10^7\log d$, $m$ is the number of points in the intersection; s/r is sampling with rejection (M2); $N$ is the number of points generated by \volesti\ per step, $W$ is the walk length; $N=\dfrac{1}{\epsilon^2}400d \log d$. Time is in seconds.
}
\end{table}

\quad
\begin{table}[!h]
\makebox[1 \textwidth][c]{       
\resizebox{0.9 \textwidth}{!}{   
\begin{tabular}{ |p{0.3cm}|p{0.8cm}|p{0.8cm}|p{1.5cm}|p{0.8cm}|p{1.5cm}|p{0.8cm}|p{1.5cm}|p{0.8cm}|p{1.5cm}|p{0.8cm}|p{0.8cm}|  }
 \hline
 \multicolumn{12}{|c|}{Experimental results for Lawrence and rejection methods.} \\
 \hline
 $d$ & $k$ & $m$ & s/r Vol & s/r Time & ex/Law Vol & ex/Law time & fl/Law Vol & fl/Law Time & Vinci Vol & Vinci Time & per. Vol  \\
 \hline
2 & $10^5$ & 969 & 0.0049 & 0.036 & 0.005 & 0.0 & 0.0005 & 0.0 & 0.005 & 0.0 & $1\%$ \\
\hline
5 & 2$\cdot 10^5$ & 2034 & 8.475e-05 & 0.056 & 8.33e-05 & 0.0 & 8.33e-05 & 0.0 & 8.33e-05 & 0.0 & $1\%$ \\
\hline
5 & 2$\cdot 10^6$ & 19967 & 8.320e-05 & 0.492 & 8.33e-05 & 0.0 & 8.33e-05 & 0.0 & 8.33e-05 & 0.0 & $1\%$ \\
\hline
10 & 3$\cdot 10^5$ & 2952 & 2.711e-09 & 0.136 & 2.76e-09 & 0.0 & 2.76e-09 & 0.0 & 2.76e-09 & 0.0 & $1\%$ \\
\hline
10 & 3$\cdot 10^6$ & 2986 & 2.743e-09 & 1.132 & 2.76e-10 & 0.0 & 2.76e-10 & 0.0 & 2.76e-10 & 0.0 & $0.1\%$ \\
\hline
15 & 3$\cdot 10^5$ & 2991 & 7.624e-15 & 0.156 & 7.64e-15 & 0.02 & 7.64e-15 & 0.0 & 7.64e-15 & 0.0 & $1\%$ \\
\hline
20 & 4$\cdot 10^5$ & 4096 & 4.209e-21 & 0.332 & 4.11e-21 & 0.052 & 4.11e-21 & 0.0 & 4.11e-21 & 0.0 & $1\%$ \\
\hline
20 & 4$\cdot 10^6$ & 39800 & 4.09e-21 & 3.204 & 4.11e-21 & 0.052 & 4.11e-21 & 0.0 & 4.11e-21 & 0.0 & $1\%$ \\
\hline
20 & 4$\cdot 10^6$ & 3894 & 4.001e-22 & 3.14 & 4.11e-22 & 0.02 & 4.11e-22 & 0.0 & 4.11e-22 & 0.0 & $0.1\%$ \\
\hline
25 & 4$\cdot 10^5$ & 4049 & 6.526e-28 & 0.416 & 6.45e-28 & 0.076 & 6.45e-28 & 0.0 & 6.45e-28 & 0.0 & $1\%$ \\
\hline
25 & 4$\cdot 10^6$ & 39858 & 6.424e-28 & 4.108 & 6.45e-28 & 0.076 & 6.45e-28 & 0.0 & 6.45e-28 & 0.0 & $1\%$ \\
\hline
30 & 4$\cdot 10^5$ & 3986 & 3.757e-35 & 0.52 & 3.77e-35 & 0.12 & 2.37e-35 & 0.0 & 3.77e-35 & 0.0 & $1\%$ \\
\hline
30 & 4$\cdot 10^6$ & 40155 & 3.785e-35 & 4.808 & 3.77e-35 & 0.12 & 3.77e-35 & 0.0 & 3.77e-35 & 0.0 & $1\%$ \\
\hline
30 & 4$\cdot 10^6$ & 3979 & 3.750e-36 & 4.96 & 3.77e-36 & 0.08 & 3.77e-35 & 0.0 & 3.77e-36 & 0.0 & $0.1\%$ \\
\hline
35 & 4$\cdot 10^5$ & 4077 & 9.864e-43 & 0.588 & 9.67e-43 & 0.184 & 9.68e-43 & 0.004 & ---- & -- & $1\%$ \\
\hline
35 & 4$\cdot 10^6$ & 40155 & 9.696e-43 & 5.852 & 9.67e-43 & 0.184 & 6.22e-42 & 0.0 & ---- & -- & $1\%$ \\
\hline
40 & 5$\cdot 10^5$ & 4977 & 1.220e-50 & 0.864 & 1.226e-50 & 0.34 & 1.06e-50 & 0.0 & ---- & -- & $1\%$ \\
\hline
40 & 5$\cdot 10^6$ & 50074 & 1.227e-50 & 8.56 & 1.226e-50 & 0.34 & 1.23e-50 & 0.0 & ---- & -- & $1\%$ \\
\hline
40 & 5$\cdot 10^6$ & 4923 & 1.207e-51 & 8.464 & 1.226e-51 & 0.344 & -1.38e-49 & 0.0 & ---- & -- & $0.1\%$ \\
\hline
50 & 5$\cdot 10^5$ & 5003 & 3.290e-67 & 1.088 & 3.28e-67 & 1.276 & 3.29e-67 & 0.0 & ---- & -- & $1\%$ \\
\hline
50 & 5$\cdot 10^6$ & 49923 & 3.283e-67 & 11.0 & 3.28e-67 & 1.276 & 2.99e-67 & 0.0 & ---- & -- & $1\%$ \\
\hline
50 & 5$\cdot 10^6$ & 5011 & 3.295e-68 & 11.068 & 3.28e-68 & 0.924 & 3.16e-68 & 0.0 & ---- & -- & $0.1\%$ \\
\hline
60 & 5$\cdot 10^5$ & 5093 & 1.224e-84 & 1.356 & 1.20e-84 & 2.6 & 3.59e-84 & 0.0 & ---- & -- & $1\%$ \\
\hline
60 & 5$\cdot 10^6$ & 50122 & 1.204e-84 & 13.5 & 1.20e-84 & 2.6 & -4.20e-83 & 0.0 & ---- & -- & $1\%$ \\
\hline
60 & 5$\cdot 10^6$ & 4897 & 1.177e-85 & 13.512 & 1.20e-84 & 2172 & -2.27e-80 & 0.0 & ---- & -- & $0.1\%$ \\
\hline
70 & 6$\cdot 10^5$ & 6069 & 8.444e-103 & 1.988 & 8.348e-103 & 5.776 & -1.85e-95 & 0.0 & ---- & -- & $1\%$ \\
\hline
70 & 6$\cdot 10^6$ & 59911 & 8.336e-103 & 19.436 & 8.348e-103 & 5.776 & -8.78e-97 & 0.0 & ---- & -- & $1\%$ \\
\hline
70 & 6$\cdot 10^6$ & 6105 & 8.494e-104 & 19.512 & 8.348e-104 & 5.048 & 9.37e-99 & 0.0 & ---- & -- & $0.1\%$ \\
\hline
70 & $\cdot 10^7$ & 10125 & 8.453e-104 & 32.208 & 8.348e-104 & 5.776 & -1.28e-95 & 0.0 & ---- & -- & $0.1\%$ \\
\hline
80 & 6$\cdot 10^5$ & 6059 & 1.410e-121 & 2,24 & 1.397e-121 & 11.564 & 2.33e-91 & 0.0 & ---- & -- & $1\%$ \\
\hline
80 & 6$\cdot 10^6$ & 59991 & 1.397e-121 & 22.576 & 1.397e-121 & 11.564 & ---- & -- & ---- & -- & $1\%$ \\
\hline
80 & 6$\cdot 10^6$ & 5965 & 1.389e-122 & 22.424 & 1.397e-121 & 11.272 & ---- & -- & ---- & -- & $0.1\%$ \\
\hline
90 & 6$\cdot 10^5$ & 6045 & 6.781e-141 & 2.492 & 6.73e-141 & 25.036 & ---- & -- & ---- & -- & $1\%$ \\
\hline
90 & 6$\cdot 10^6$ & 59873 & 6.717e-141 & 24.384 & 6.73e-141 & 25.036 & ---- & -- & ---- & -- & $1\%$ \\
\hline
90 & 6$\cdot 10^6$ & 6083 & 6.823e-142 & 24.416 & 6.73e-142 & 20.764 & ---- & -- & ---- & -- & $0.1\%$ \\
\hline
90 & $\cdot 10^7$ & 10036 & 6.755e-142 & 41.232 & 6.73e-142 & 25.3 & ---- & -- & ---- & -- & $0.1\%$ \\
\hline
100 & 6$\cdot 10^5$ & 6020 & 1.075e-160 & 2.696 & 1.072e-160 & 41.56 & ---- & -- & ---- & -- & $1\%$ \\
\hline
100 & 6$\cdot 10^6$ & 60190 & 1.075e-160 & 27.096 & 1.072e-160 & 41.56 & ---- & -- & ---- & -- & $1\%$ \\
\hline
100 & 6$\cdot 10^6$ & 6034 & 1.077e-161 & 27.472 & 1.072e-161 & 37.352 & ---- & -- & ---- & -- & $0.1\%$ \\
\hline
100 & $\cdot 10^7$ & 9979 & 1.069e-161 & 45.168 & 1.072e-161 & 33.612 & ---- & -- & ---- & -- & $0.1\%$ \\
\hline
\end{tabular}
} } 
\caption{\label{fig:text3} $k$ is the number of points sample from the unit simplex, $k=10^x\log d$, with $x=max\{5, 4+\lceil -\log_{10}(p)\rceil \}$, where $p$ is the percentage of the unit simplex volume of the defined polytope, $m$ is the number of points in the intersection and last column is $p$. Time is in seconds.}
\end{table}

\quad
\begin{table}[!h] 
\makebox[1 \textwidth][c]{ \resizebox{1.3 \textwidth}{!}{   
\begin{tabular}{|p{0.3cm}||p{0.8cm}|p{0.8cm}|p{1.8cm}|p{1.0cm}||p{0.8cm}|p{0.8cm}|p{0.3cm}|p{1.8cm}|p{1.0cm}|p{1.8cm}|p{1.0cm}|  }
 \hline
 \multicolumn{12}{|c|}{Experimental results for the unit simplex and ellipsoid intersection.} \\ \hline
 $d$ & $k$ & $m$ & s/r Vol. & s/r Time & $N$ & $W$ & $\epsilon$ &
  s/V Vol. & s/V Time & o/V Vol. & o/V Time \\ \hline
3 & $10^5$ & 1318 & 0.0804667 & 0.004 & 14648 & 10 & 0.3 & 0.0792319 & 0.592 & 0.0798146 & 0.564\\
\hline
6 & $10^5$ & 7668 & 0.001065 & 0.004 & 47780 & 10 & 0.3 & 0.00107003 & 14.172 & 0.00105103 & 13.412\\
\hline
8 & $10^5$ & 8798 & 2.18204e-05 & 0.012 & 73935 & 10 & 0.3 & 2.18847e-05 & 48.672 & 2.22077e-05 & 55.324 \\
\hline
15 & 2$\cdot 10^5$ & 19827 & 7.58102e-13 & 0.02 & 64993 & 11 & 0.5 & 7.68531e-13 & 96.888 & 7.46954e-13 & 105.06 \\
\hline
20 & 3$\cdot 10^5$ & 29951 & 4.1036e-19 & 0.036 & 66571 & 12 & 0.5 & 3.93709e-19 & 178.54 & 3.97954e-19 & 170.476 \\
\hline
25 & 3$\cdot 10^5$ & 39987 & 6.44486e-26 & 0.056 & 89413 & 12 & 0.6 & 6.4879e-26 & 457.196 & 6.51637e-26 & 442.68 \\
\hline
30 & 3$\cdot 10^5$ & 39987 & 3.76754e-33 & 0.056 & 63772 & 14 & 0.8 & 3.92896e-33 & 311.772 & 3.56866e-33 & 311.872 \\
\hline
40 & 4$\cdot 10^5$ & 39974 & 1.22559e-48 & 0.096 & 40987 & 14 & 1.2 & 1.21068e-48 & 253.38 & 1.30804e-48 & 242.172 \\
\hline
40 & 4$\cdot 10^5$ & 39999 & 1.22559e-48 & 0.096 & 59022 & 14 & 1.0 & 1.28713e-48 & 436.976 & 1.26529e-48 & 448.472 \\
\hline
\end{tabular}
} } 
\caption{\label{fig:text4} $k$ is the number of points sampled in the simplex, of which $m$ lie in the intersection; $N$ is the number of points generated by \volesti\ per step, $W$ is the walk length; $N=\dfrac{1}{\epsilon^2}400d \log d$. Time is in seconds.}
\end{table} 

\quad
\begin{table}[!h] 
\makebox[1 \textwidth][c]{ \resizebox{1.0 \textwidth}{!}{   
\begin{tabular}{|p{0.3cm}||p{0.8cm}|p{0.8cm}|p{1.8cm}|p{1.0cm}||p{0.8cm}|p{0.8cm}|p{0.3cm}|p{1.8cm}|p{1.0cm}|  }
 \hline
 \multicolumn{10}{|c|}{Experimental results for non convex bodies.} \\ \hline
 $d$ & $k$ & $m$ & s/r Vol. & s/r Time & $N$ & $W$ & $\epsilon$ &
  s/V Vol. & s/V Time \\ \hline
3 & $5\cdot 10^5$ & 6384 & 0.00213 & 0.172 & 14648 & 10 & 0.3 & 0.00174 & 2.028 \\
\hline
6 & $5\cdot 10^5$ & 43210 & 0.000120 & 0.22 & 47780 & 10 & 0.3 & 0.000120 & 22.036 \\
\hline
8 & $5\cdot 10^5$ & 72915 & 3.616e-06 & 0.26 & 73935 & 10 & 0.3 & 3.633e-06 & 114.596 \\
\hline
15 & 5$\cdot 10^5$ & 38012 & 5.814e-14 & 0.448 & 64993 & 11 & 0.5 & 5.834e-14 & 139.908 \\
\hline
15 & 5$\cdot 10^5$ & 41824 & 6.044e-14 & 0.476 & 64993 & 12 & 0.5 & 8.109e-14 & 240.74 \\
\hline
20 & 5$\cdot 10^5$ & 31824 & 2.616e-20 & 0.644 & 95863 & 12 & 0.5 & 2.642e-20 & 1016.15 \\
\hline
20 & 5$\cdot 10^5$ & 36273 & 2.981e-20 & 0.620 & 66571 & 12 & 0.6 & 2.895e-20 & 323.536 \\
\hline
25 & 5$\cdot 10^5$ & 27650 & 3.565e-27 & 0.86 & 89413 & 12 & 0.6 & 3.787e-27 & 999.352 \\
\hline
25 & 5$\cdot 10^5$ & 27055 & 3.488e-27 & 0.82 & 65691 & 12 & 0.7 & 3.301e-27 & 586.496 \\
\hline
30 & 5$\cdot 10^5$ & 26451 & 1.994e-34 & 1.032 & 83294 & 13 & 0.7 & 2.171e-34 & 1051.19 \\
\hline
30 & 5$\cdot 10^5$ & 26265 & 1.980e-34 & 1.072 & 83294 & 13 & 0.7 & 2.179e-34 & 1005.43 \\
\hline
35 & 5$\cdot 10^5$ & 2158 & 4.176e-43 & 1.196 & 49774 & 14 & 1.0 & 2.904e-44 & 630.908 \\
\hline
35 & 5$\cdot 10^5$ & 1115 & 2.158e-43 & 1.348 & 61450 & 13 & 0.9 & 1.198e-166 & 1417.01 \\
\hline
35 & 5$\cdot 10^5$ & 10160 & 1.966e-42 & 1.292 & 61450 & 13 & 0.9 & 1.061e-42 & 810.248 \\
\hline
40 & 5$\cdot 10^5$ & 8753 & 1.22559e-48 & 1.36 & 72866 & 13 & 0.9 & 2.087e-192 & 1873.56 \\
\hline
\end{tabular}
} } 
\caption{\label{fig:text5} $k$ is the number of points sampled in the simplex, set to constant $k=5\cdot 10^5$, of which $m$ lie in the intersection; $N$ is the number of points generated by \volesti\ per step, $W$ is the walk length. We set $N=\dfrac{1}{\epsilon^2}400d \log d$. Time is in seconds.}
\end{table}

\quad
\newpage
\section{Applications} 

\begin{figure}[h!]
		\includegraphics[width=1.2\textwidth]{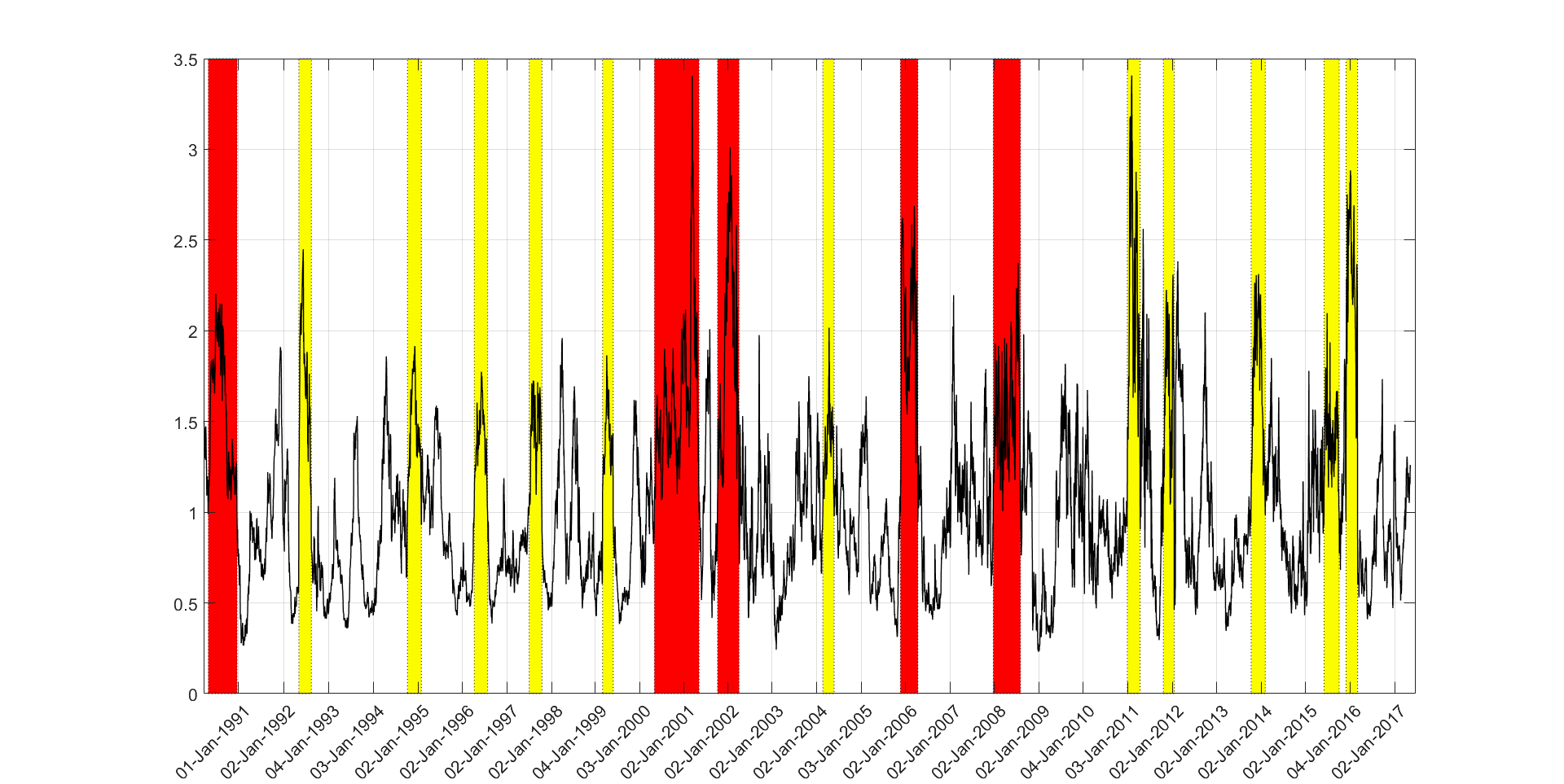}
	\caption{Representation of the periods over which the indicator is greater than one for 61-100 days (yellow) and over 100 days (red)}
	\label{fig:WarningsCrises}
\end{figure}

\begin{table}[h!]
	\centering
	\begin{tabular}{llc}
	Start date & End date & Duration (days)\\
		\hline
		02-May-1990 & 20-Dec-1990 & 166 \\
		06-May-1992 & 14-Aug-1992 & 72 \\
		06-Oct-1994 & 27-Jan-1995 & 80 \\
		08-Apr-1996 & 24-Jul-1996 & 77 \\
		01-Jul-1997 & 13-Oct-1997 & 74 \\
		03-Mar-1999 & 01-Jun-1999 & 61 \\
		04-May-2000 & 09-May-2001 & 258 \\
		05-Oct-2001 & 05-Apr-2002 & 124 \\
		25-Feb-2004 & 28-May-2004 & 65 \\
		18-Nov-2005 & 11-Apr-2006 & 101 \\
		20-Dec-2007 & 04-Aug-2008 & 157 \\
		28-Dec-2010 & 12-Apr-2011 & 75 \\
		18-Oct-2011 & 16-Jan-2012 & 63 \\
		08-Oct-2013 & 04-Feb-2014 & 82 \\
		04-Jun-2015 & 05-Oct-2015 & 87 \\
		30-Nov-2015 & 03-Mar-2016 & 66 \\ 
	\end{tabular}
	\caption{All periods over which the return/volatility indicator is greater than one for more than 60 days.}
	\label{tab:Warnings}
\end{table}

\begin{figure}[h!]
		\includegraphics[width=1.2\textwidth]{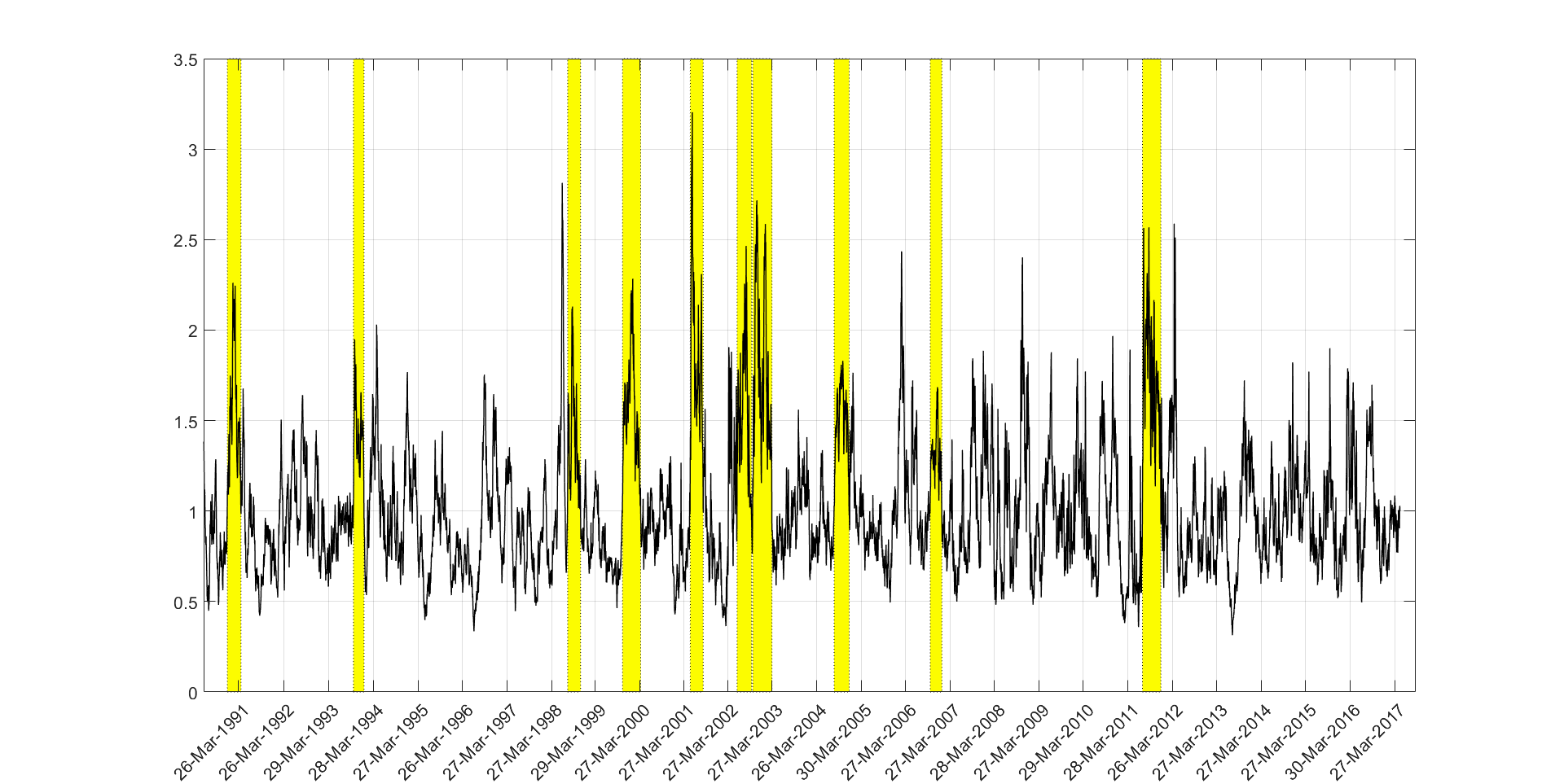}
	\caption{Representation of the periods over which the indicator is greater than one for over 60 days (yellow)}
	\label{fig:Momentum}
\end{figure}

\begin{table}[h!]
	\centering
	\begin{tabular}{llc}
	Start date & End date & Duration (days)\\
		\hline
		26-Dec-1990 & 16-Apr-1991 & 79 \\
		18-Oct-1993 & 11-Jan-1994 & 61 \\
		11-Aug-1998 & 24-Nov-1998 & 75 \\
		08-Nov-1999 & 04-Apr-2000 & 105 \\ 
		22-May-2001 & 04-Sep-2001 & 75 \\
		14-Jun-2002 & 09-Oct-2002 & 83 \\
		18-Oct-2002 & 27-Mar-2003 & 111 \\ 
		20-Aug-2004 & 21-Dec-2004 & 87 \\
		13-Oct-2006 & 19-Jan-2007 & 67 \\
		26-Jul-2011 & 21-Dec-2011 & 106 \\ 
	\end{tabular}
\caption{All periods over which the momentum indicator is greater than one for more than 60 days.\label{tab:Momentum}}
\end{table}

\end{document}